\documentclass{amsart}

\usepackage{graphicx}
\usepackage{amssymb,amsmath}
\usepackage{amsthm,hyperref}

\newtheorem{lemma}{Lemma}[section]
\newtheorem{proposition}[lemma]{Proposition}
\newtheorem{theorem}[lemma]{Theorem}
\theoremstyle{definition}

\newtheorem*{question*}{Question}
\newtheorem{remark}[lemma]{Remark}
\newtheorem{example}[lemma]{Example}

\DeclareMathOperator{\Var}{Var}

\title{Multi-target detection with  rotations}

\author[T. Bendory]{Tamir Bendory}
\address{School of Electrical Engineering, Tel Aviv University, Tel Aviv,
Israel}

\author[T.Y. Lan]{Ti-Yen Lan}
\address{Program in Applied and Computational Mathematics, Princeton University,
Princeton, NJ, USA}

\author[N.F. Marshall]{Nicholas F. Marshall}
\address{Department of Mathematics, Oregon State University, Corvallis, OR, USA}

\author[I. Rukshin]{Iris Rukshin}
\address{Program in Applied and Computational Mathematics, Princeton
University, Princeton, NJ, USA}

\author[A. Singer]{Amit Singer}
\address{Program in Applied and Computational Mathematics and the Department of
Mathematics, Princeton University, Princeton, NJ, USA}

\keywords{Multi-target detection,  autocorrelation analysis, bispectrum,
single-particle reconstruction, cryo-EM.}

\begin{document}

\begin{abstract}
We consider the multi-target detection problem of 
estimating a two-dimensional target image from a
large noisy measurement image that contains many randomly rotated and translated
copies of the target image. Motivated by single-particle cryo-electron
microscopy, we focus on the low signal-to-noise regime, where
it is difficult to estimate the locations and orientations of the target images
in the measurement. Our approach uses autocorrelation analysis to estimate 
rotationally and translationally invariant features of the target image. We
demonstrate that, regardless of the level of noise, our technique can be used to
recover the target image when the measurement is sufficiently large. 
\end{abstract}

\maketitle

\section{Introduction} \label{intro}
Let $M$ be a noisy measurement image that contains $p$ randomly rotated and
translated copies of a target image $f$. More precisely, suppose that $f :
\mathbb{R}^2 \rightarrow \mathbb{R}$ is supported on the unit disc, and $f_\phi$
is the rotation of $f$ by angle $\phi$ about the origin.  Further, let $F_\phi :
\mathbb{Z}^2 \rightarrow \mathbb{R}$ be the discretization of $f_\phi$ defined
by $F_\phi({x}) = f_\phi({x}/n)$ for a fixed integer $n$. We assume that
the measurement $M : \{1,\ldots,m\}^2 \rightarrow \mathbb{R}$ has the form
\begin{equation} \label{Meq}
M({x}) = \sum_{j=1}^p F_{\phi_j}(x-x_j) +
\varepsilon({x}),
\end{equation}'.
where $\phi_1,\ldots,\phi_p \in [0,2\pi)$ are uniformly random rotations;
${x}_1,\ldots,{x}_p \in \{n+1,\ldots,m-n\}^2$ are arbitrary
translations; and $\varepsilon({x})$ is i.i.d.\ Gaussian noise on $\{1,\ldots,m\}^2$
with mean zero and variance $\sigma^2$, see the example in Figure \ref{fig:M}.

We further impose a separation condition
$|{x}_{j_1} - {x}_{j_2}| \ge 4n$ for $j_1 \not = j_2$, which ensures that
the targets in the measurement are separated by at least the diameter of their
support. We also assume a density condition $pn^2/m^2 := \gamma > 0$ so that the
targets appear in the measurement at some minimal density.
{ Moreover, it is necessary to assume that $f$ has some regularity;
we assume $f$ is bandlimited (in the harmonics on the disc); see
\ref{bandlimit2d}.}
\begin{figure}[ht!]
\centering
\includegraphics[width=.9\textwidth]{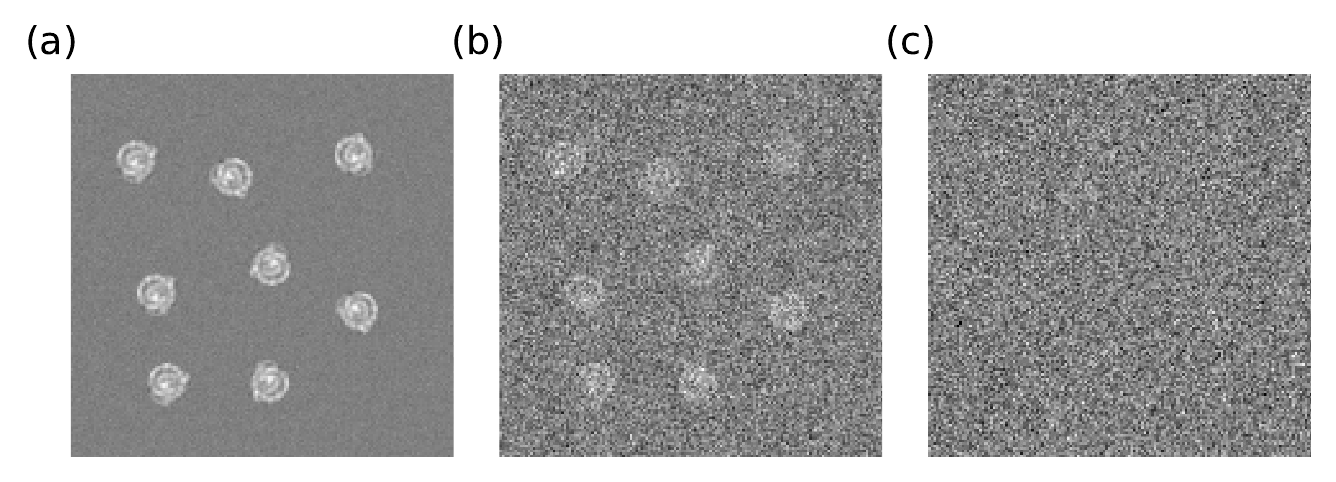} 
\caption{An example of the measurement $M$ defined in \eqref{Meq} with (a)
SNR = $10^2$, (b) SNR = 1 and (c) SNR := $10^{-2}$, where $\text{SNR}=(\pi
n^2\sigma^{2})^{-1} \sum_x F(x)^2$ .} \label{fig:M}
\end{figure}

Given the measurement $M$, the objective is to recover the function $f$. This
problem is called multi-target detection (MTD) with
rotations~\cite{bendory2019multi,lan2019multi}.  Motivated by
single-particle cryo-electron microscopy (cryo-EM), we focus on the low SNR
regime, see Figure~\ref{fig:M}(c), where estimating the unknown translations
and rotations is challenging~\cite{aguerrebere2016fundamental,
bendory2018toward, henderson1995potential}.
We pose the following question.

\begin{question*}
Suppose that $M : \{1,\ldots,m\}^2 \rightarrow \mathbb{R}$ is a measurement of
the form described in~\eqref{Meq} for fixed signal radius $n$ and
density $\gamma$.  If the variance of the noise~$\sigma^2$ is fixed (but might
be arbitrarily high), can the function $f$ be estimated from~$M$ to any fixed
level of accuracy when $m$ is sufficiently large?
\end{question*}

In this paper, we develop a mathematical and computational  framework for
MTD with rotations, and show empirically that the answer to the above
question is affirmative. In particular, we describe an { autocorrelation analysis} algorithm for
recovering the function~$f$ from the measurement $M$ and demonstrate its
effectiveness and numerical stability. Additionally, in Section~\ref{sec:1d}, we consider a simplified
version of this statistical estimation problem in one dimension, where we are
able to establish a theoretical foundation for this algorithm.

%

\section{Motivation and related work} \label{sec:motivation}
\subsection{Motivation}
Our interest in the MTD model arises from the structure determination 
problem for biological molecules. In the past decade, cryo-EM has emerged as a
potent alternative to X-ray crystallography and nuclear magnetic resonance (NMR) spectroscopy 
to resolve the structures of proteins that either cannot be crystallized or are too complex for NMR. 
In cryo-EM, a solution that contains many copies of the target particle is rapidly cooled 
to form thin vitreous ice sheets whose thickness is comparable to the single molecule size.
These sheets are then imaged with an electron microscope.
The measurements in cryo-EM can be modeled as two-dimensional
tomographic projections of identical biomolecules at unknown locations and
orientations followed by some image distortion due to the imaging system.
The projection images are embedded in a large, noisy image, called a micrograph.
The crux of single-particle cryo-EM reconstruction is that, with sufficiently
many micrographs, projection images of similar molecule orientations 
can be combined to improve the SNR and, in turn, reconstruct the 
high-resolution three-dimensional structure of the molecule.

The current computational pipeline for cryo-EM requires particle picking, the
extraction of the biomolecule projection images from the micrographs~\cite{chen2007signature,
eldar2020klt, heimowitz2018apple, scheres2015semi,
wang2016deeppicker,wagner2019sphire}. Then, the three-dimensional structure is
built from the extracted images using a variety of
algorithms~\cite{bendory2019single, frank2006three, grant2018cistem,
punjani2017cryosparc, scheres2012relion, tang2007eman2}. This approach is
problematic for small particles where the SNR of micrographs is low, and
detection becomes impossible~\cite{bendory2018toward, henderson1995potential,aguerrebere2016fundamental}.
As such, the difficulty of detection sets a lower bound on the usable molecule
size in the current analysis workflow of cryo-EM data.

Interest in signal recovery beyond the detection limit has prompted the
realization that the locations of the signal in the measurement are nuisance
parameters; the emerging claim is that signal recovery can be achieved directly
from the measurement~\cite{bendory2018toward}. Methodologies for direct image
estimation have been inspired by Zvi Kam's introduction of autocorrelation
analysis to the structure reconstruction problem dealing with randomly 
oriented biomolecule projections~\cite{kam1980reconstruction}. 
The process involves accumulating the ``spatial correlations", or autocorrelations, of signal density in the
measurements in order to average out the noise without estimating the rotations.
These averages are then used for the reconstruction of the target image. Following
Kam's seminal paper on autocorrelations, several procedures based on
correlations and moments have been proposed for cryo-EM and related modalities,
e.g.,
~\cite{bandeira2017estimation,bendory2017bispectrum,abbe2018multireference,levin20183d,perry2019sample,saldin2011reconstructing,saldin2010structure,sharon2020method,huang2022orthogonal}.


{
\begin{remark}[Relation of model of this paper to cryo-EM]
The model \eqref{Meq} considered in this paper involves a large noisy
measurement $M$ that contains many instances of a 2D target image at arbitrary
locations and random orientations. This model is a simplified version of cryo-EM
data that, informally speaking, consists of a large noisy measurement that
contains many tomographic projections of a 3D density at arbitrary locations and
random orientations. While the 2D model we study is not directly applicable to
cryo-EM data,  it does represent a step towards understanding the application of
invariant feature based approaches for cyro-EM by building upon past
work on multi-target detection
\cite{bendory2019multi,kreymer2021two, kreymer2021approximate,lan2019multi,
shalit2021generalized}. Moreover, the model considered in this paper
corresponds to a degenerate case in cryo-EM in which the molecule has a
preferred orientation. Random conical tilt \cite{radermacher1987three} is a
classical reconstruction method in cryo-EM that assumes a preferred
orientation. The model considered in this paper has recently been extended to
random conical tilt \cite{lan2021a}  which does have direct potential
applications.
\end{remark}
}

\subsection{Related work}
The problem addressed in this paper---with rotated and translated iterations
of~$f$ within~$M$---extends previous works on the MTD
model~\cite{bendory2019multi,lan2019multi}. In particular, we extend~\cite{marshall2020image} by providing new theoretical understanding of a
1-dimensional model, and demonstrating empirically that reconstruction is
possible from a measurement $M$ of the form~\eqref{Meq}. This is an
important step toward the reconstruction of molecules in
the undetectable domain. More generally, it attests to the possibility of direct
image estimation
from measurements so that limitations on particle picking do not
necessarily translate to limitations on structure determination.

We mention that our results were recently extended, after this paper appeared
online, to account for an arbitrary distribution of the target
images~\cite{kreymer2021two}. In addition,  an approximate
expectation-maximization algorithm for the MTD model with rotations was
developed in~\cite{lan2019multi,kreymer2021approximate}, and a generalized method of moments
framework was designed in~\cite{shalit2021generalized}.

\section{One-Dimensional Problem} \label{sec:1d}
Before considering the two-dimensional problem~\eqref{Meq}, we introduce
an analogous problem in one dimension. This simplified version will allow us to
develop intuition for the autocorrelation framework we
devise for the two-dimensional case.

\subsection{Measurement} \label{statement1d} Let $F :
\mathbb{Z} \rightarrow \mathbb{R}$ be a one-dimensional target signal supported on $\{-n,\ldots,n-1\}$, and
$F_\tau : \mathbb{Z} \rightarrow \mathbb{R}$ be the result of cyclically rotating the support of $F$. 
That is, $F_\tau(x) = F((x +\tau) \bmod 2 n)$ for $x
\in \{-n,\ldots,n-1\}$, where we consider an integer modulo $2n$ to be an
element of $\{-n,\ldots,n-1\}$, and $F_\tau(x) = 0$ when $x \in \mathbb{Z}
\setminus \{-n,\ldots,n-1\}$. 
Here, we will work with { a} one-dimensional analogue  of the
two-dimensional measurement defined in \eqref{Meq}, 
where the measurement~$M : \{1,\ldots,m\} \rightarrow \mathbb{R}$ is given by
\begin{equation}
\label{M1d}
M(x) = \sum_{j=1}^p F_{\tau_j}(x - x_j) + \varepsilon(x), 
\end{equation}
where $\tau_1,\ldots,\tau_p \in \{-n,\ldots,n-1\}$ are uniformly random
cyclic shifts; $x_1,\ldots, x_p \in \{n+1,\ldots, m-n+1\}$ are arbitrary
translations;
and $\varepsilon$ is i.i.d.\ Gaussian noise on $\{1, \ldots, m\}$ with mean zero and
variance $\sigma^2$. 
We plot an examples of the 1-dimensional measure $M$ with three different
levels of noise in Figure \ref{fig1D}.

\begin{figure}[ht!]
\centering
\begin{tabular}{ccc}
\includegraphics[width=.29\textwidth]{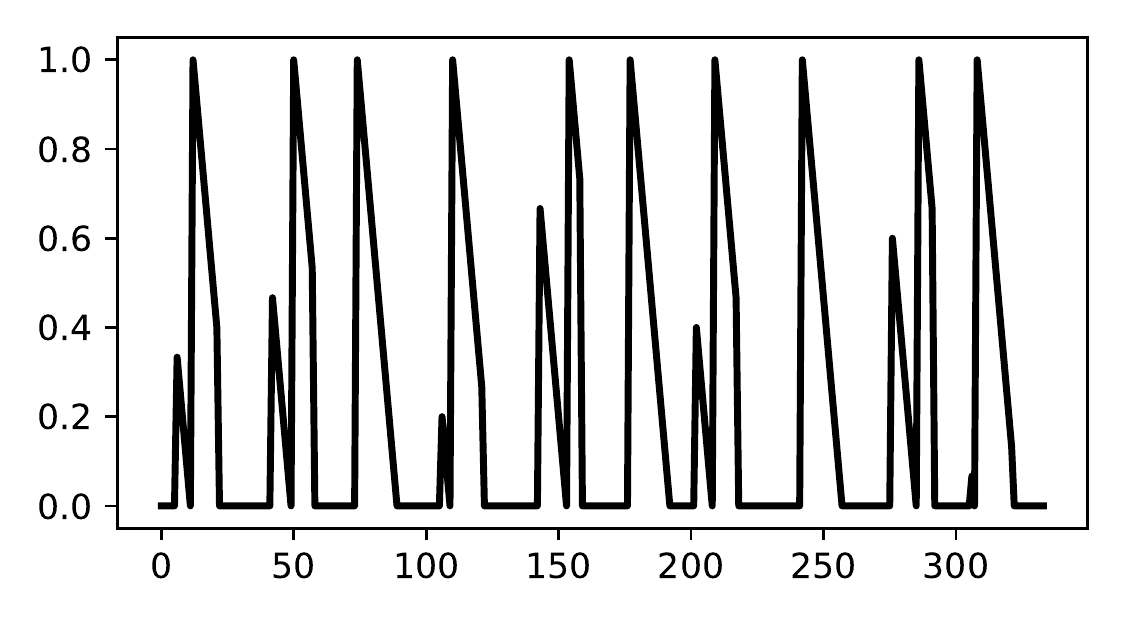} &
\includegraphics[width=.29\textwidth]{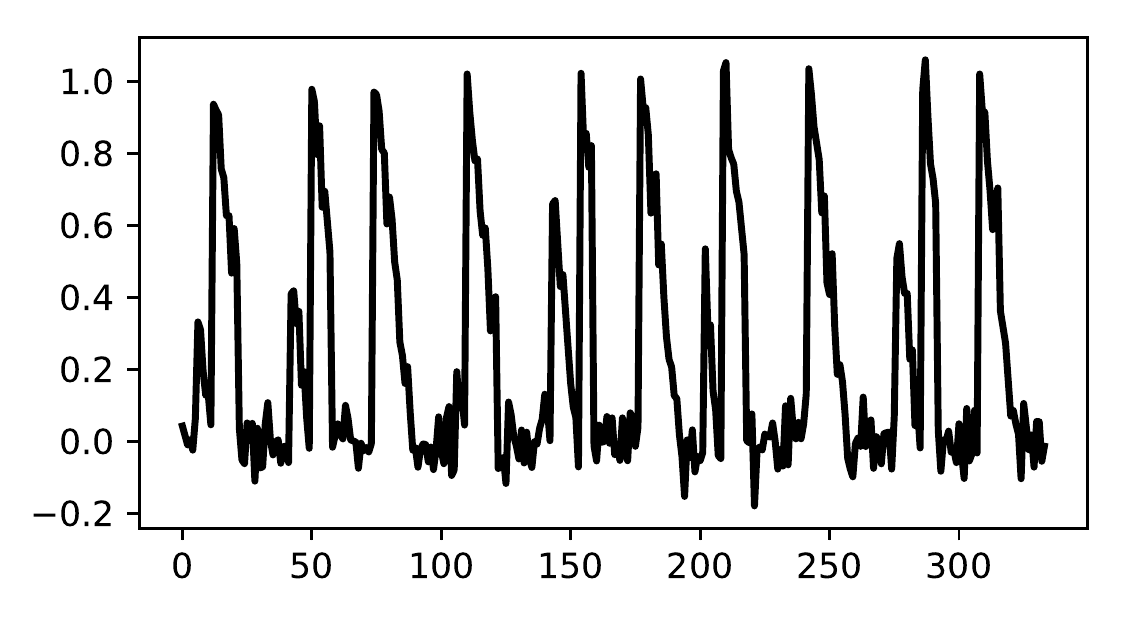} &
\includegraphics[width=.29\textwidth]{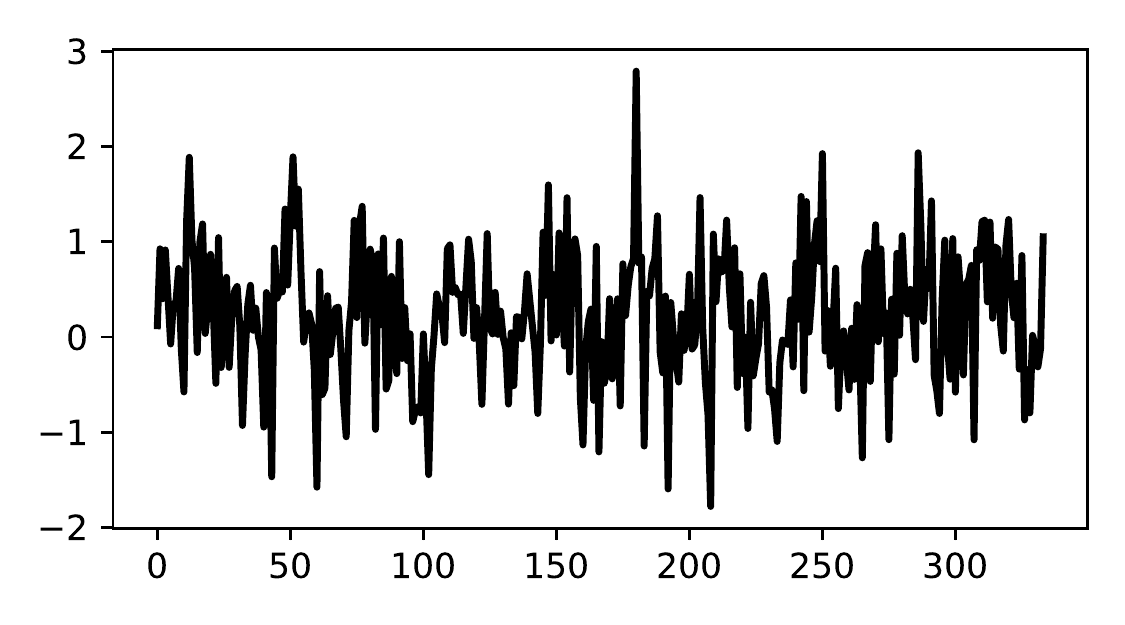} 
\end{tabular}
\caption{An example of the 1-dimensional measurement $M$ defined in \eqref{M1d} 
with $\text{SNR}=\infty$ (left), $\text{SNR}=10^2$ (middle), and
$\text{SNR}=1$ (right), where $\text{SNR}:= {(2} n{)}^{-1} \sum_x
F(x)^2/\sigma^2$.} \label{fig1D}
\end{figure}

The circular shifts $\tau$ are the one-dimensional analogues of the
two-dimensional rotations $\phi$ in~\eqref{Meq}. To extend the
previous assumptions about the target separation and bounded density to this
one-dimensional formulation, we assume that $|x_j - x_k|
\ge 4n$ for all $k \not = j$ and that $n p /m = \gamma > 0$. In Theorem~\ref{autos}, we will additionally impose that the
discrete Fourier transform (DFT) of $F$ is non-vanishing. 

As above, our objective is to estimate the function $F$ from the measurement~$M$
in the low SNR regime. In particular, we would like to show that $F$ can be
reliably and accurately estimated from $M$ at any fixed level of noise, which
might be arbitrarily high, as long as the size of the measurement $m$ is
sufficiently large. The reliability and accuracy of this estimate will be
quantified below. The approach is based on seeking features of $F$ that
determine the function and are invariant to translations $F(x) \mapsto F(x +
x')$ and circular shifts of the support $F(x) \mapsto F_\tau(x)$. The
construction of these invariant features is based on autocorrelation analysis.

\subsection{Invariant features} \label{invariant1d}
Recall that $F : \mathbb{Z} \rightarrow \mathbb{R}$ is supported on
$\{-n,\ldots,n-1\}$ and $F_\tau$ is a rotated version of $F$. We can define
features of $F$ that are invariant to rotations and translations. The
most direct example is the mean of the function 
\begin{equation} \label{TF}
T_F = \frac{1}{2n} \sum_{x = -n}^{n - 1} F(x).
\end{equation}
Motivated by autocorrelation analysis, the mean above can also be interpreted as
the first-order autocorrelation. The rotationally-averaged second-order
autocorrelation $U_F: \{-2n, \ldots, 2n-1\} \to \mathbb{R}$ is defined
by
\begin{equation*}
U_F(x_1) = \frac{1}{2n} \sum_{\tau = -n}^{n - 1} \frac{1}{2n} \sum_{x = -n}^{n - 1} F_{\tau}(x) F_{\tau}(x + x_1).
\end{equation*}
Considering the sum geometrically { (or by a change of variables)  } we observe that $U_F(x_1)$ is only a function
of the magnitude $|x_1|$, and so, cannot contain sufficient information to
recover $F$.  Thus, the critical invariant is the rotationally-averaged
third-order autocorrelation $V_F :
\{-2n,\ldots,2n-1\}^2 \rightarrow \mathbb{R}$, which is defined by
\begin{equation} \label{Teq}
V_F(x_1,x_2) = \frac{1}{2n} \sum_{\tau = -n}^{n-1} \frac{1}{2n} \sum_{x = -n}^{n - 1} 
F_\tau(x) F_\tau(x + x_1) F_\tau(x + x_2).
\end{equation}
By construction, both $U_F(x_1)$ and $V_F(x_1, x_2)$ are invariant under
translations $F(x)$ $\mapsto F(x-x')$ and rotations $F(x) \mapsto F_{\tau'}(x)$.
That is, $U_F = U_G$ and $V_F = V_G$ 
when $G(x) = F_{\tau'}(x + x')$ for any $\tau' \in \{-n,\ldots,n-1\}$ and
$x' \in \mathbb{Z}$.  

\subsection{Estimation from measurement} \label{measurement1d}
The function $V_F : \{-2n,\ldots,2n-1\}^2 \rightarrow \mathbb{R}$ can be estimated
from a measurement $M : \{1,\ldots,m\} \rightarrow \mathbb{R}$ of the form
described in \S \ref{statement1d}. For simplicity, let us extend the separation
condition so that it also holds periodically in the sense that $|x_{k_1} -
x_{k_2} - m| > 4n$.  We define the third-order autocorrelation of the
measurement $A_M : \{-2n\ldots,2n-1\}^2
\rightarrow \mathbb{R}$ by
\begin{equation} \label{amauto3}
A_M(x_1,x_2) = \\ \frac{1}{m} \sum_{x =1}^m M(x) M(x+x_1 \bmod m) M(x+x_2 \bmod m),
\end{equation}
where an integer modulo $m$ is taken as an element of $\{1,\ldots,m\}$.  

The following lemma shows that $V_F$ can be estimated from $A_M$ (namely, from
the data) if $m$ is much larger than $\sigma^6$. Information theoretic results
that were derived for a closely related model called multi-reference alignment
indicate that this is the optimal estimation rate in the low SNR regime where
$m,\sigma\to\infty$ while $\gamma$ and $n$ are
fixed~\cite{abbe2018estimation,bandeira2017estimation,perry2019sample}. 

\begin{lemma} \label{lem1}
Suppose that $|F| < F_{max}$ everywhere for some constant $F_{max}>0$. Under the
one-dimensional model~\eqref{M1d}, we have: 
$$
    \mathbb{E}( A_M(x_1,x_2)) = \frac{\gamma}{n} V_F(x_1, x_2) +  2 \gamma T_F
\sigma^2 (\delta_0 (x_1 - x_2) + \delta_0 (x_1) + \delta_0 (x_2)),
$$
and
\begin{equation*}
\Var \left( A_M(x_1,x_2) \right) = \mathcal{O} \left(\frac{n}{m} \left(
\gamma F_{max}^6 + \sigma^6 \right) \right),
\end{equation*}
where the expectation and variance are taken with respect to the random cyclic
shifts and the Gaussian noise, and $\delta_0(x) = 1$ when $x = 0$ and $\delta_0(x)
=0$ otherwise.	
\begin{proof}
	See Appendix~\ref{sec:technical_lemmas}.	
\end{proof}

\end{lemma}

The retrieval of $F$ from $V_F$, combined with Lemma \ref{lem1}, would result in the extraction of $F$ from the measurement $M$ up to a rotation, given a sufficiently large measurement.

\subsection{Recovery from invariant features}
The DFT of the function $F : \mathbb{Z}
\rightarrow \mathbb{R}$ considered as a function on its support
$\{-n,\ldots,n-1\}$ is defined by
\begin{equation} \label{aeq}
	a_k := \sum_{x=-n}^{n-1} F(x) e^{-2\pi i k x/(2n)}, \quad k\in \{-n,\ldots,n-1\}.
\end{equation}
We now show that $V_F$ determines $F$ via a closed form when its  DFT is non-vanishing. 
We remark that such
a non-vanishing condition is standard for problems related to autocorrelation
inversion, see for example~\cite{bendory2017bispectrum,perry2019sample}.

\begin{theorem} \label{autos}
Suppose that the DFT of $F$ expressed in \eqref{aeq} is non-vanishing. Then,
$\tilde{F}$ can be determined from $V_F$ via a closed form expression
{ (resulting from inverting a linear system only depending on $n$)}
such that $\tilde{F} = F_\tau$ for some $\tau \in \{-n,\ldots,n-1\}$. That is,
$F$ can be recovered up to a circular shift.
\end{theorem}

\begin{proof}
Let $A_{F} : \{-2n,\ldots,2n-1\}^2 \rightarrow \mathbb{R}$ designate the third-order autocorrelation
\begin{equation} \label{eq_auto}
A_F(x_1, x_2) =\\ \frac{1}{2n} \sum_{x = -2n}^{2n-1} F(x) F((x + x_1)\bmod 4n) F((x + x_2) \bmod
4n ),
\end{equation}
where an integer modulo $4n$ is taken to be element of
$\{-2n,\ldots,2n-1\}$. Observe that we have
$$
V_F(x_1,x_2) = \frac{1}{2n} \sum_{\tau = -n}^{n-1} A_{F_\tau}(x_1,x_2).
$$
Indeed, since $F$ is supported on $\{-n,\ldots,n-1\}$, taking it as a periodic
function in \eqref{eq_auto} does not change the result. Let $b_m$ denote the
Fourier coefficients of $F$ considered as a periodic function on
$\{-2n,\ldots,2n - 1\}$; that is,
$$
b_m := \sum_{x=-n}^{n-1} F(x) e^{-2 \pi i m x /(4n)}.
$$
By Fourier inversion on the interval $\{-2n,\ldots,2n-1\}$, we have
\begin{equation} \label{ffinv}
F(x) = \frac{1}{4n} \sum_{m=-2n}^{2n-1} b_m e^{2\pi i m x/(4n)}.
\end{equation}
Substituting the representation of $F(x)$ in terms of the coefficients $b_m$
into $A_F$ and summing over $x$ yields 
$$
A_F(x_1,x_2) = \frac{1}{2n} \frac{1}{(4n)^3} \cdot \\ \sum_{m_1,m_2 = -2n}^{2n-1} b_{m_1} b_{m_2} b_{-m_1-m_2} e^{2\pi i (m_1 x_1 + m_2 x_2)/(4n)}.
$$
Next, by taking the two-dimensional DFT of $A_F(x_1,x_2)$ on
$\{-2n,\ldots,2n-1\}^2$, we can recover $b_{m_1} b_{m_2} b_{-m_1-m_2}$ for $m_1,
m_2 \in \{-2n,\ldots,2n-1\}$:
$$
b_m = \sum_{j=-n}^{n - 1}\left(\frac{1}{2n} \sum_{k=-n}^{n-1} a_k e^{2\pi i k
j/(2n)} \right) e^{-2\pi i m j/(4n)} \\ = \sum_{k=-n}^{n-1} \gamma_{m,k} a_k,
$$
where 
$$
\gamma_{m,k} := \frac{1}{2n} \sum_{j=-n}^{n-1} e^{2\pi i kj/(2n)}e^{-2\pi imj/(4n)}.
$$
Observe that if $m = 2k'$ we have $\gamma_{m,k'} = 1$ and $\gamma_{m,k} = 0$ if
$k \not = k'$. It follows that
$$
b_{2k_1} b_{2k_2} b_{-2k_1-2k_2} = a_{k_1} a_{k_2} a_{-k_1-k_2}
$$
for $k_1,k_2 \in \{-n,\ldots,n-1\}$. The quantity $a_{k_1} a_{k_2} a_{-k_1-k_2}$, called the bispectrum of the function $F$,  is invariant under cyclic shifts
of the underlying function $F$ and determines $F$ uniquely, up to a global cyclic shift ~\cite{bendory2017bispectrum,sadler1992shift,tukey1953spectral}. 
{
More precisely, by taking the logarithm of the bispectrum we arrive at a linear
system of equations which is full rank after the cyclic shift ambiguity is
removed, see \cite[\S IV.C]{bendory2017bispectrum}. Therefore, for each $n$,
there is a fixed linear transform that determines the Fourier coefficients of
$F$, up to a phase ambiguity. Since the reduction of $V_F$ to the bispectrum can
also be accomplished by a linear transform, for any fixed $n$ composing these
linear transformations together with Fourier inversion gives a closed form
expression for determining $F$ up to cyclic shift from $V_F$.
}
Thus, given $V_F$, we can determine $\tilde{F} = F_{\tau}$ for some $\tau \in
\{-n, \ldots, n-1\}$, as desired.
\end{proof}
{
\begin{example}
To illustrate the approximation result of Lemma~\ref{lem1} and Theorem~\ref{autos}, we present a basic numerical example.  We use the signal from Figure
\ref{fig1D} with $\text{SNR} = 10^2$. Next, we form a measurement of the form~\eqref{M1d} with various numbers of samples $p$ of the given function. We use
the identities described in the proof of Theorem \ref{autos} to approximate the
bispectrum $a_{k_1} a_{k_2} a_{-k_1 -k_2}$ for $k_1,k_2 \in \{-n,\ldots,n-1\}$
of the given signal $F$. We plot the relative error of the bispectrum extracted from
the measurement compared to the ground truth, see Figure \ref{fignew}. The error decreases as $1/\sqrt{p}$, as expected by the law of large numbers.
The
function can be recovered from the bispectrum using a variety of standard
methods, see  \cite{bendory2017bispectrum}.
\begin{figure}[ht!]
\centering
\includegraphics[width=.38\textwidth]{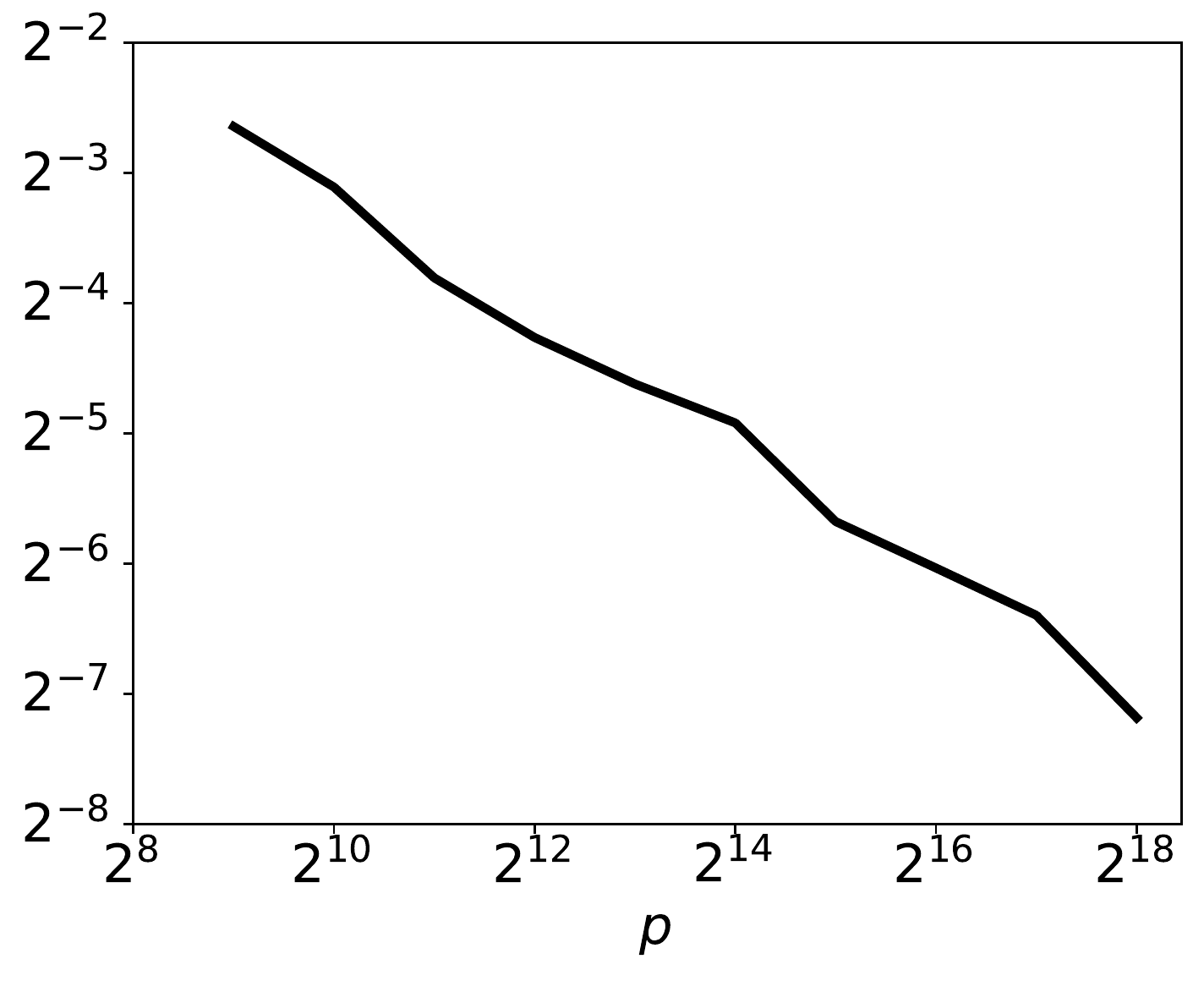} 
\caption{Relative error of bispectrum derived from measurement 
$M$ using various
number of samples $p$ averaged over 10 trials. 
The error decreases as $1/\sqrt{p}$, as expected by the law of large numbers.}
\label{fignew} 
\end{figure}
\end{example}
}
{
\begin{remark}[Discretization model]
In the above model, we consider a function~$F$ defined on a gird that is
transformed by on grid translations. One potential extension of this model is to
consider off  grid translations by assuming that $F$ represents samples from an
underlying function $f$ defined on the real-line;
more precisely, the model
\eqref{M1d} could be extended by introducing a shift parameter $\delta \in
[0,1/n)$  and defining
$$
F^\delta(x) = f(x/n + \delta),
$$
for a discretization parameter $n \in \mathbb{Z}_{>0}$
and an underlying function $f : \mathbb{R} \rightarrow \mathbb{R}$ that is
supported on $[-1,1]$. With this notation, the measurement model
\eqref{M1d} could be extended by defining $M : \{1,\ldots,m\} \rightarrow
\mathbb{R}$ by 
$$
M(x) = \sum_{j=1}^p F_{\tau_j}^{\delta_j}(x - x_j) + \varepsilon(x),
$$
where $\tau_j \in \{-n,\ldots,n-1\}$ is a random cyclic rotation, $x_j \in
\{n+1,\ldots,m-n+1\}$ are arbitrary translations, $\varepsilon$ is i.i.d.
Gaussian noise on $\{1,\ldots,m\}$, and $\delta_j \in [0,1/n]$ is a random shift.
Under this model, the third-order autocorrelation $A_M(x_1,x_2)$ defined in
\eqref{amauto3} would satisfy an analogous version of Lemma~\ref{lem1},
where the features $T_F$ and $V_F$ are replaced by quantities defined with
appropriate integrals instead of sums. Studying this extended model would
require quantifying an additional source of error when trying to determine $f$  
from its third-order autocorrelation; in particular, it would be necessary to 
make assumptions justifying why the Fourier inversion formula \eqref{ffinv} 
approximately holds (for example, one could assume that $f$ is Lipchitz
continuous, or assume that $f$ and its derivatives are Lipchitz continuous up
to order $k$ such that classical approximation theory results could be
employed). 

In this paper, we focus on translations on the discretization grid
(for both the 1D and 2D models we consider) to avoid dealing with this
additional source of approximation error. Our goal in considering on grid
translations is to study the simplest possible model that still captures the
essence of the signal processing problem of interest: MTD in a setting where
there are two different types of random linear actions.
 Extending the results of this
paper to handle arbitrary shifts would be a necessary extension if the presented
approaches are adapted for an application problem involving real data.
\end{remark}
}

\section{Two-dimensional problem}

After having established the theoretical foundation for the one-dimensional
problem above, we aim to extend the recovery of the underlying function $f$ to
two dimensions. In order to make this estimation tractable, it is necessary to
make regularity assumptions on the function $f$; we build the foundation for 
these assumptions below.

\subsection{Invariant features in the continuous setting}
As in \S \ref{invariant1d}, we define the continuous two-dimensional analogues for the features of $f$ that are invariant under
translations and rotations. As before, the first invariant is the
mean of the function
$$
q_f := \int_{\mathbb{R}^2} f(x) dx.
$$
Letting $f_\phi$ be the rotation of $f$ by angle $\phi$ about the origin, we define rotationally-averaged second-order autocorrelation
$r_f : \mathbb{R}^2 \rightarrow \mathbb{R}$ by
$$
r_f(x_1) :=  \frac{1}{2\pi} \int_0^{2\pi} \int_{\mathbb{R}^2}
f_\phi(x) f_\phi(x +x_1) d x  d \phi.
$$
Finally, to gain enough information for the recovery of $f$, the rotationally-averaged third-order autocorrelation $s_f:
\mathbb{R}^2 \times \mathbb{R}^2 \rightarrow \mathbb{R}$ is
$$
s_f(x_1,x_2) := \frac{1}{2\pi} \int_0^{2 \pi} \int_{\mathbb{R}^2} 
f_\phi(x) f_\phi(x+x_1) f_\phi(x+x_2) dx d\phi. 
$$
In this case, observe that $s_f$ is a function of $|x_1|,|x_2|$ and the angle
$\theta(x_1,x_2)$ between~$x_1$ and $x_2$. Geometrically, as a function of three 
variables, $s_f$ potentially contains enough information to recover $f$.

\subsection{Invariant features in the discrete setting} 
As in the one-dimensional case, we will focus on the recovery of some
discretization of $f$ from some discretization of~$s_f$ under similar
assumptions to those in Theorem~\ref{autos}. Restricting our attention
to this problem is consistent with the fact that actual measurements are
discretized over a pixel grid. Conveniently, this also considerably simplifies
the presentation of the method.

We define the discretization $F_\phi : \mathbb{Z}^2 \rightarrow \mathbb{R}$ of
$f_\phi$ by
$$
F_\phi(x) = f_\phi(x/n), \quad \text{for} \quad x \in \mathbb{Z}^2,
$$
where $n$ is a fixed integer that determines the sampling resolution. We
define the discrete rotationally-averaged third-order autocorrelation $S_f :
\mathbb{Z}^2 \times \mathbb{Z}^2 \rightarrow \mathbb{R}$ by
\begin{equation} \label{SF}
S_F(x_1,x_2) := \\ \frac{1}{2\pi} \int_0^{2\pi} \frac{1}{4 n^2}  \sum_{x \in \mathbb{Z}^2}
F_\phi(x) F_\phi(x + x_1) F_{\phi}(x+x_2) d\phi.
\end{equation}
Since $f_\phi$ is supported on the open unit disc $\{x \in \mathbb{R}^2: |x| <
1\}$, it follows that $F_\phi$ is supported on $\{ x \in \mathbb{Z}^2 : |x| < n\}$, and
$S_F(x_1,x_2)$ is supported on 
$$
\mathcal{X} := \{-2n,\ldots,2n-1\}^2 \subset \mathbb{Z}^2,
$$
which contains $(2n)^2$ points. 

\subsection{Estimation from measurement} \label{2dautodef}
Suppose that $M : \{1,\ldots,m\}^2 \rightarrow \mathbb{R}$, 
a measurement of the form in \eqref{Meq}, is given.
We define the third-order autocorrelation of  $M$ as $A_M : \mathbb{Z}^2 \times
\mathbb{Z}^2 \rightarrow \mathbb{R}$ by
$$
A_M({x}_1, {x}_2) := \frac{1}{m^2} \sum_{{x} \in \mathbb{Z}^2} 
M({x}) M({x}+{x}_1) M({x}+{x}_2).
$$
Recall that the measurement $M : \{1,\ldots,m\}^2 \rightarrow \mathbb{R}$ is
defined by
$$
M({x}) = \sum_{j=1}^p F_{\phi_j}(x-x_j) + 
\varepsilon({x}),
$$
where $\phi_j$ are random angles, $x_j$ are translations, and $\varepsilon$ is
noise (see \S \ref{intro}). As before, we assume that images in $M$ are
separated by at least one image diameter according to
$$
| {x}_{j_1} - {x}_{j_2} | \ge 4n, \quad \text{for} \quad j_1 \not =
j_2,
$$
and that the density of the target images in the measurement is $p n^2/m^2 = \gamma$ for a
fixed constant $\gamma > 0$. Under these assumptions, it is straightforward to
show that for any fixed level of noise $\sigma^2$, fixed signal radius
$n$ and fixed $\gamma$,
\begin{equation}
\label{eq:A2S}
A_M({x}_1, {x}_2) \rightarrow \frac{\gamma}{2\pi} S_F({x}_1, {x}_2) + \\ \frac{\gamma}{2\pi} \sigma^2
\mu_F \big( \delta({x_1}) + \delta({x_2}) + \delta({x_1} - {x_2}) \big),
\end{equation}
as $m \rightarrow \infty$ (see for example~\cite{bendory2018toward}), where $\mu_F$ is the discrete mean of $F$ defined by
$$
\mu_F = \frac{1}{4n^2} \sum_{x \in \mathbb{Z}^2} F_{\phi}(x).
$$
As such, \eqref{eq:A2S} relates the third-order autocorrelation of the measurement
$A_M$ to the invariant features $S_F$ and $\mu_F$. In practice, $\sigma^2$ and $\gamma \mu_F$
can be estimated from $M$: $\sigma^2$ can be estimated by the variance of the
pixel values of $M$ in the low SNR regime, while $\gamma \mu_F$ can be estimated
by the empirical mean of $M$. As a result, $S_F$, a feature of the image, can be estimated from $A_M$, a feature of the measurement, up to a constant factor.

\subsection{Band-limited functions on the unit disc} \label{bandlimit2d}
The Dirichlet Laplacian eigenfunctions
on the unit disc $D = \{ (x,y) \in
\mathbb{R} : x^2 + y^2 \le 1 \}$ are solutions to the eigenvalue problem
$$
\left\{  \begin{array}{cc}
-\Delta \psi = \lambda \psi & \text{in } D \\
\psi = 0 & \text{on } \partial D,
\end{array}
\right.
$$
where $-\Delta = -(\partial_{x x} + \partial_{y y})$ is the Laplacian, and
$\partial D$ is the boundary of the unit disc. In polar coordinates
$(r,\theta)$, these eigenfunctions are of the form 
\begin{equation} \label{eq1}
\psi_{\nu,q}(r,\theta) = J_{\nu}\left( \lambda_{\nu, q} r \right) e^{i \nu \theta},
\end{equation}
where $\nu \in \mathbb{Z}_{\ge 0}$, $J_{\nu}$ is the $\nu$-th order Bessel function of
the first kind, and $\lambda_{n,q} > 0$ is the $q$-th positive root of $J_{\nu}$.
Recall that $J_{\nu}$ is a solution to the differential
equation 
$$
y''(r) +  \frac{1}{r} y'(r) + \left(1 - \frac{\nu^2}{r^2} \right) y(r) = 0.
$$
Therefore, by writing the Laplacian $-\Delta$ in polar coordinates, we have
$$
-\Delta \psi_{\nu,q}(r,\theta) = -\left( \partial_{rr} + \frac{1}{r} \partial_r
+ \frac{1}{r^2} \partial_{\theta \theta} \right) \psi_{\nu,q}(r,\theta) = 
\lambda_{\nu,q}^2 \psi_{\nu,q}(r,\theta),
$$
and, as such, $\lambda_{\nu,q}^2$ is the eigenvalue corresponding to the
eigenfunction $\psi_{\nu,q}$. Therefore, the projection operator 
$$
P_\lambda f = \sum_{(\nu,q) : \lambda_{\nu,q} \le \lambda} \frac{\langle f,
\psi_{\nu,q} \rangle}{\|\psi_{\nu,q}\|_{2}^2} \psi_{\nu,q}
$$
can be viewed as a low-pass filter for functions on the unit disc; we call
functions  that are invariant under this projection operator band-limited { functions}.

\subsection{Steerable bases} \label{steerable}
Recall that $f : \mathbb{R}^2 \rightarrow \mathbb{R}$ is supported on the unit
disc. Using the notation from \S \ref{bandlimit2d}, the assumption that $f$ is
band-limited on its support can be written as 
\begin{equation}
\label{model}
f(r,\theta) =	\sum_{(\nu,q): \lambda_{\nu,q} \le \lambda} \alpha_{\nu,q} \psi_{\nu,q}(r,\theta), \quad \text{for } r \le 1,
\end{equation}
where $\lambda > 0$ is the band-limit frequency, and $\alpha_{\nu,q}$ are
expansion coefficients. For each $\nu$, we define
$$
g_\nu(r,\theta) = \sum_{q : \lambda_{\nu,q} \le \lambda} 
	\alpha_{\nu,q} \psi_{\nu,q}(r,\theta) = \\ \left( \sum_{q : \lambda_{\nu,q} \le \lambda}
\alpha_{\nu,q} J_{\nu}\left(
\lambda_{\nu,q} r \right)\right) e^{i \nu \theta},
$$
so that we can write $f$ by
\begin{equation} \label{sumg}
f(r,\theta) = \sum_{\nu = -N}^{N}
g_\nu(r,\theta), 
\end{equation}
where $N := \max \{ \nu : \lambda_{\nu,1} \le \lambda\}$. 

The advantage of expressing a function in terms of Dirichlet Laplacian
eigenfunctions is that the basis is steerable---the effect of rotations on
expansion coefficients of the images are expressed as phase modulation.
Specifically, a steerable basis diagonalizes the rotation operator so
that the rotation $f_\phi(r,\theta) :=
f(r,\theta+\phi)$ of $f$ about the origin by angle $\phi$ can be computed
by multiplying each term in the sum in \eqref{sumg} by 
 $e^{i \nu \phi}$:
\begin{equation} \label{grot}
f_\phi(r,\theta) = \sum_{\nu = -N}^{N} g_\nu(r,\theta) e^{i \nu \phi}.
\end{equation}

From this point forward, we will switch between considering functions in polar coordinates
$f(r,\theta)$ or Cartesian coordinates $f({x})$, where 
${x} = (r \cos \theta, r \sin \theta )$, 
depending on which is more
convenient.

\subsection{Using the band-limited assumption} \label{bandlimit}
We now take advantage of the assumption that $f$ is band-limited on the unit
disc.  Let $\Psi_{\nu, q}: \mathcal{X} \rightarrow \mathbb{C}$ be the
discretization of the Dirichlet Laplacian eigenfunctions
$$
\Psi_{\nu, q}(x) = \psi_{\nu,q}(x/n),
$$ 
where $\psi_{\nu,q}$ is supported on the unit disc, as in \S \ref{bandlimit2d}.
With this notation, 
$$
F_\phi(x) = \sum_{(\nu,q): \lambda_{\nu,q} \le \lambda} \alpha_{\nu,q}
	\Psi_{\nu,q}(x) e^{i \nu \phi}.
$$
By considering $F_\phi$ and $S_F$ as functions on
$\mathcal{X}$, we can express their DFT 
$\hat{F}_\phi : \mathcal{X} \rightarrow \mathbb{C}$  by
$$
\hat{F}_\phi(k) = \sum_{x \in \mathcal{X}}
F_\phi(x) e^{-2\pi i  x\cdot k /(4n)}.
$$
Finally, we let $\hat{\Psi}_{\nu, q}: \mathcal{X} \rightarrow \mathbb{C}$ be the DFT of $\Psi_{\nu, q} : \mathcal{X} \rightarrow \mathbb{C}$
$$
\hat{\Psi}_{\nu,q}(k) = \sum_{x \in \mathcal{X}} \Psi_{\nu,q}(x) e^{-2\pi i
x\cdot k /(4n)}.
$$
Then, by the linearity of the DFT  it follows from the previous section that
$$
\hat{F}_\phi(k) =	\sum_{(\nu,q) \in \mathcal{V}}
\alpha_{\nu,q} \hat{\Psi}_{\nu, q}(k) e^{i \nu \phi},
$$
where $\mathcal{V} = \{ (\nu,q) : \lambda_{\nu,q} \le \lambda \}$.

\subsection{Discrete Fourier transform of invariant features}
The Fourier transform defined in the previous section can now be related to the Fourier transform of $S_F$. We define $\hat{S}_F : \mathcal{X} \times \mathcal{X} \rightarrow \mathbb{C}$ by
\begin{equation} \label{SFhat}
\hat{S}_F(k_1,k_2) :=  \sum_{x_1 \in \mathcal{X}} \sum_{x_2 \in \mathcal{X}}
S_F(x_1,x_2) e^{-2\pi i (k_1 \cdot x_1 + k_2 \cdot x_2)/ 4n},
\end{equation}
where addition is considered modulo $4n$ with $-2n, \ldots, 2n-1$ as the
representatives of the different equivalence classes. Substituting \eqref{SF}
into \eqref{SFhat} and simplifying gives
$$
\hat{S}_F(k_1,k_2)
= \int_0^{2\pi} \hat{F}_\phi(k_1) \hat{F}_\phi(k_2)
\hat{F}_\phi(-k_1-k_2) d\phi. 
$$
This integral over $\phi$ can be
replaced by a summation over the rotations at the Nyquist rate so that the expression becomes:
\begin{equation} \label{SFE}
\hat{S}_F(k_1,k_2) 
= \sum_{j=0}^{6N-1}
\hat{F}_{\phi_j}(k_1) \hat{F}_{\phi_j}(k_2)
\hat{F}_{\phi_j}(-k_1-k_2),
\end{equation}
where $\phi_j := 2\pi j/(6 N)$. 
 If $N := \max \{ \nu : \lambda_{\nu,1} \le \lambda\}$, then the
products $F_\phi(x) F_\phi(x + x_1) F_\phi(x +
x_2)$ that appear in~\eqref{SFE} are band-limited by $ N$
with respect to $\phi$.
Also, note that the summand on the right hand side of
\eqref{SFE} is the DFT of the third autocorrelation of a
function and is called the bispectrum \cite{bendory2017bispectrum,sadler1992shift,tukey1953spectral}. We encountered the one-dimensional
analogue of the bispectrum at end of the proof to Proposition~\ref{autos}.


\subsection{Vector notation} \label{vecnote}
Recall the enumeration $\mathcal{V}$ from \S \ref{bandlimit}, which consists of 
$(\nu_1,q_1),\ldots,(\nu_d,q_d)$ such that 
$$
\hat{F}_\phi(k) = \sum_{j=1}^d \alpha_{\nu_j,q_j} \hat{\Psi}_{\nu_j,q_j}(k)
e^{i \nu_j \phi},
$$
for $k \in \mathcal{X}$. For a fixed angle $\phi$ and $k \in \mathcal{X}$, we
define the vector $u(\phi,k) \in \mathbb{R}^d$ by
$$
u_j(\phi,k) = \hat{\Psi}_{\nu_j,q_j}(k) e^{i \nu_j \phi}.
$$
Thus, each vector $v \in \mathbb{R}^d$ defines the DFT of
a band-limited function  by
$$
\hat{F}_{v,\phi}(k) = \sum_{j=1}^n v_j u_j(\phi,k) = v^\top u(\phi,k),
$$
where $v^\top$ is the transpose of $v$.  The following lemma is immediate from
\eqref{SFE} and the product rule.

\begin{lemma} \label{lemcomp}
We have
$$
\hat{S}_{F_v}({k}_1,{k}_2) = \\ \sum_{j=0}^{6N-1} 
v^\top u(\phi_j,k_1) 
v^\top u(\phi_j,k_2)
v^\top u( \phi_j, -{k}_1-{k}_2) 
$$
where $\phi_j := 2\pi j/(6 N)$. Moreover, the $d$-dimensional gradient $\nabla_v
\hat{S}_{F_v}$ satisfies
\begin{multline*}
\nabla_v \hat{S}_{F_v}(k_1,k_2) =  \sum_{j=0}^{6N-1} \Big( v^\top
u(\phi_j,k_1)  v^\top u(\phi_j,k_2) u(\phi_j,-k_1-k_2) \\
+ v^\top u(\phi_j,k_1) v^\top u(\phi_j,-k_1-k_2) u(\phi_j,k_2) \\
+ v^\top u(\phi_j,k_2) v^\top u(\phi_j,-k_1-k_2) u(\phi_j,k_1) \Big).
\end{multline*}
\end{lemma}
In the next section, we describe how to estimate $S_F$ from a measurement
$M$ and form an optimization problem to recover the target image $F$
from the measurement~$M$.

\subsection{Computational complexity}
For some intuition regarding the computational complexity of computing
$\hat{S}_{F_v}$ and $\nabla_v S_{F_v}$, recall that $\mathcal{X} =
\{-2n,\ldots,2n-1\}^2$ so that, crudely, the image $F$ has
$\sim n^2$ pixels. In the following, we make the assumption that
the number of eigenfunctions used to expand the function $F$ should not exceed
the number of pixels in the image. Notationally, $|\mathcal{V}| = \mathcal{O}(n^2)$. 

\begin{proposition}  
We can compute $\hat{S}_{F_v}(k_1,k_2)$ and $\nabla_v \hat{S}_{F_v}(k_1,k_2)$
for all~$(k_1,k_2) \in \mathcal{X}^2$ in~$\mathcal{O}(n^5)$ operations.
\end{proposition}

\begin{proof} First, we compute $v^\top u(\phi_j,k)$ for $j = 0,\ldots,6N-1$ and
$k \in \mathcal{X}$.  Each inner product resolves to $\mathcal{O}(n^2)$
operations for $\mathcal{O}(N n^2)$ total evaluations. 
Eigenvalue asymptotics show that $N$ is of the order of $\sqrt{|\mathcal{V}|} =
\mathcal{O}(n)$ so this computation involves $\mathcal{O}(n^5)$ operations.

After this pre-computation, it is straightforward to calculate
$\hat{S}_{F_v}(k_1,k_2)$ for all $(k_1,k_2) \in \mathcal{X}^2$ in
$\mathcal{O}(n^5)$ operations.  For $\nabla_v \hat{S}_{F_v}(k_1,k_2)$, the
key observation is that the gradient is a linear combination of
$\mathcal{O}(n^3)$ vectors of length $\mathcal{O}(n^2)$. First, we compute the
coefficients in $\mathcal{O}(n^5)$ operations and then sum the vectors in
$\mathcal{O}(n^5)$ operations.

\end{proof}

\section{ Algorithms and Numerical Results}
\subsection{Optimization problem}
We now delineate an optimization problem for the estimation of $F$ from $S_F$. 
Recall that each vector $v \in \mathbb{R}^d$ defines the Fourier transform of a
band-limited function on the disc by 
$$
\hat{F}_{v,\phi}(k) =  v^\top u(\phi,k),
$$
as in \S \ref{vecnote}. We define the least squares cost function $g : \mathbb{R}^d \rightarrow
\mathbb{R}$ by
$$
g(v) = \frac{1}{2} \sum_{(k_1,k_2) \in \mathcal{X}^2} \left(
\hat{S}_{F_v}(k_1,k_2) - \hat{S}_F(k_1,k_2) \right)^2.
$$
Using the chain rule, the gradient of $\nabla g$ is
$$
\nabla g (v)
= \\ \sum_{(k_1,k_2) \in \mathcal{X}^2}  \left( \hat{S}_{F_v}(k_1,k_2) -
\hat{S}_F(k_1,k_2) \right) \nabla_v \hat{S}_{F_v}(k_1,k_2) ,
$$
where $\hat{S}_{F_v}$ and $\nabla_v \hat{S}_{F_v}$ can be computed via the
formulas in Lemma \ref{lemcomp}.


\subsection{Recovery from invariant features}
Given a cost function and a gradient, there are a variety of optimization
methods that can be used. For simplicity, we use the Broyden-Fletcher-Goldfarb-Shanno (BFGS)
algorithm, which is a popular gradient-based optimization method. 

First, we consider the problem of recovering
$F$ from $S_F$ in the absence of noise. We generate a band-limited image $F$ by
projecting a $65 \times 65$ image of a tiger onto the span of the first $600$
Dirichlet Laplacian eigenfunctions as in Figure \ref{fig1}.
\begin{figure}[ht!]
\centering
\includegraphics[width=.5\textwidth]{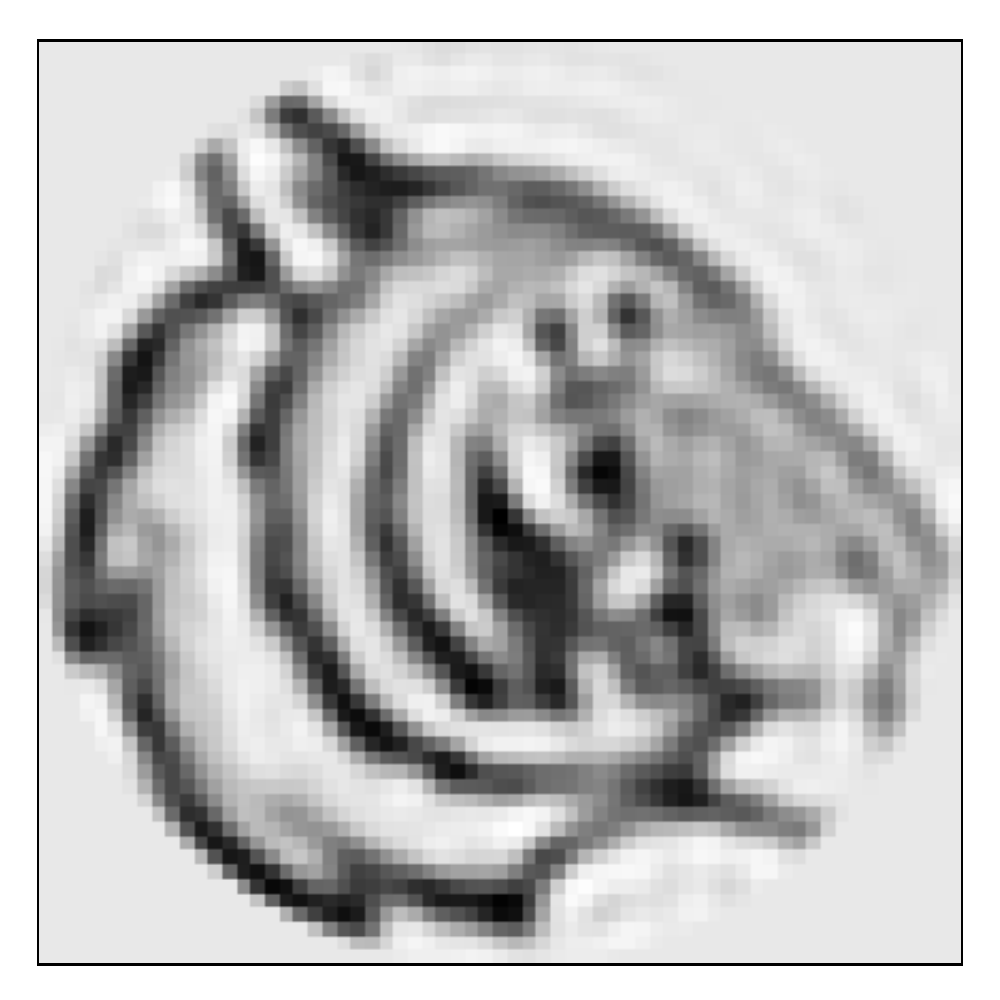} 
\vspace{-1ex}
\caption{\label{fig1} The projection of a $65 \times 65$ image of a tiger onto
the span of the first $600$ Dirichlet Laplacian eigenfunctions on a disc.}
\end{figure}
Using the BFGS optimization algorithm, the image in Figure \ref{fig1} can be recovered with 
reconstruction error $\text{error}_\text{recon} = 5\times 10^{-12}$. 
This optimization takes $6.5\times10^4$ seconds parallelized over 100 CPUs in total. 

{
\begin{remark}[Computational limitations of implementation]
We use 600 eigenfunctions for this example due to computational limitations. The
method was implemented as a CPU code, but is highly amiable to parallelization.
Implementing the method described in this paper to support the use of GPUs would
greatly increase the number of images and eigenfunctions that could be
considered; however, since our goal is to present a proof-of-concept of the
method, we choose not to pursue this optimization of the code for our numerical
results; however, creating a GPU version of the method is an interesting
potential extension of this work.
\end{remark}
}

\subsection{Using symmetry to average noise} \label{binstrat}
Recall that $S_F(x_1,x_2)$ is a discrete version of $s_f: \mathbb{R}^2 \times
\mathbb{R}^2 \rightarrow \mathbb{R}$ by
$$
s_f(x_1,x_2) := \frac{1}{2\pi} \int_0^{2 \pi} \int_{\mathbb{R}^2}
f_\phi(x) f_\phi(x+x_1) f_\phi(x+x_2) dx d\phi,
$$
which only depends on the three parameters: the magnitudes $|{x}_1|$, $|{x}_2|$
and the angle $\theta({x}_1,{x}_2)$ between $x_1$ and $x_2$.  Moreover, the
Fourier transform $\hat{s}_f$ of $s_f$ will have these same symmetries. So, it
follows that $\hat{S}_F$, which is a discrete version of $\hat{s}_f$, will also
approximately exhibit these symmetries. 

However, since $S_F$ is sampled on a grid, the symmetry will not be exact.
In order to still take advantage of the expected symmetry when $S_F(x_1,x_2)$ is estimated from a noisy measurement $M$, we introduce a ``binning''
function. Let $b :\mathcal{X} \times \mathcal{X} \rightarrow \mathbb{Z}^3$ be defined by
$$
b({k}_1,{k}_2) = \left( \left\lfloor b_1|{k}_1|
\right\rfloor, \left\lfloor b_1|{k}_2| \right\rfloor
,\left\lfloor b_2 \theta({k}_1,{k}_2) \right\rfloor \right),
$$
for fixed parameters $b_1,b_2 \in \mathbb{R}$ and $\mathcal{T} \subset
\mathbb{Z}^3$ be the range of $b$. 
The corresponding cost function $g_b : \mathbb{R}^d \rightarrow
\mathbb{R}$ is then
$$
g_b(v) = \frac{1}{2} \sum_{T \in \mathcal{T}} \left( \sum_{(k_1,k_2) \in
I_T} \left( \hat{S}_{F_v}(k_1,k_2) -
\hat{S}_F(k_1,k_2) \right) \right)^2,
$$
such that
$$
\nabla g_b (v)
= \sum_{T \in \mathcal{T}}  \left( 
\sum_{(k_1,k_2) \in  I_T}
\hat{S}_{F_v}(k_1,k_2) -
\hat{S}_F(k_1,k_2) \right) \cdot 
\sum_{(k_1,k_2) \in  I_T}
\nabla
\hat{S}_{F_v}(k_1,k_2) ,
$$
where $I_T =\{(k_1,k_2) \in \mathcal{X}^2 : b(k_1,k_2) = T\}$ is the
pre-image of $T$ under $b$. This is the same as the cost function $g(v)$ above
except that the elements are now summed within the same symmetric bin $I_T$. For
the numerical results reconstructing $F$ from a noisy measure, we report errors
with binning.

{
\begin{remark}[Estimating noise level] In the following numerical example, we assume that the noise level is known.
However, for applications it would be necessary to estimate the noise level.
In standard cryo-EM experiments, estimating the statistics of the noise is part of the standard computational pipeline~\cite{bendory2019single}.
 If
the noise dominates the signal---which is the regime of interest of this paper---then it may be easy to get a good initial guess
of the noise level; see for example~\cite{eldar2020klt}.
 Afterwards, the noise could be estimated iteratively, or the
algorithm could be run at various noise levels.

\end{remark}
}

{
\begin{remark}[Uniform in-plane rotation]
 In cryo-EM, the in-plane rotations are uniformly distributed since the specific rotation of the micrograph is arbitrarily chosen by the practitioner and there is no
 physics reason for bio-molecules to prefer a given orientation.
In contrast, the viewing directions distribution is typically non-uniform due to several physic reasons, such as specimen adherence to the air-water interface~\cite{baldwin2020non}.
 
%
\end{remark}
}

\subsection{Recovery from noisy micrographs}
Returning to the original problem presented in \S \ref{intro}, we
recall the problem of recovering a band-limited image $f$ from a measurement $M$ 
as the size $m$ of the measurement tends to infinity. To approximate the behavior practically, for 
both computational purposes and reflecting the applications to cryo-EM, we
fix $m = 1000$ and call a $1000 \times 1000$
measurement $M$ a micrograph. { The numerical experiments follow
the model described in this paper:  we define each $1000
\times 1000$ micrograph by \eqref{Meq}. We compute the third order
autocorrelation $A_M$ as defined by \S \ref{2dautodef}; we assume that
$\gamma$ and $\sigma$ are known for simplicity. }

We report numerical results in terms of the
number of independent micrographs that are used. In particular, we report both
the relative error in the invariant $S_F$ and the recovery of the image $F$. The
error for $S_F$ is summed over the bins while the relative error in recovering
$F$ is calculated after running BFGS optimization using the cost and gradient
described in \S \ref{binstrat}. In numerical experiments, we take a target image
with $n = 17$ and assume that the image is band-limited in the first $100$
Dirichlet Laplacian eigenfunctions. The relative errors from these experiments
are displayed in Figure~\ref{results}.

\begin{figure}[ht!] 
\centering
\includegraphics[width=0.6\textwidth]{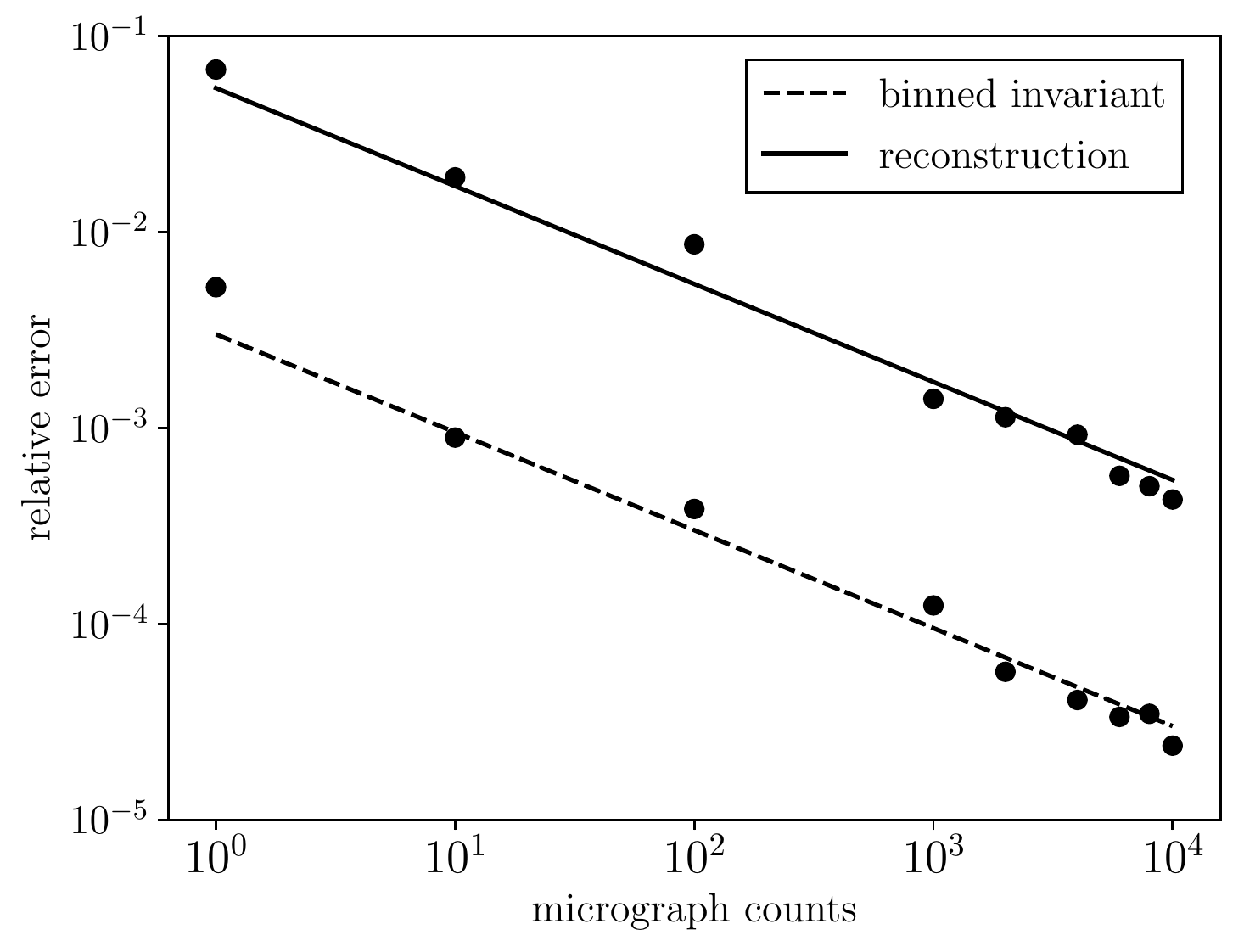}  \caption{\label{results} Relative error in the binned invariant $S_F$ summed over the bins 
and the reconstruction of $F$ as a function of the number of micrographs for 
an image with $n = 17$ with SNR $= 10^2$.}
\end{figure}

Note that the relative error in the binned invariant $S_F$ and reconstruction
both decrease at a consistent rate of one over the square root of the number of
micrographs (the expected estimation rate if the locations and rotations of the
images were known). Thus, with enough micrographs, these results indicate that
the recovery is possible regardless of the level of noise. Moreover, this
conclusion resolves the initial question of this paper: we have presented an
algorithm for the recovery of an image $f$ from a measurement of the form
\eqref{Meq} that gives predictable results in terms of the error in the
invariant.

\section{Discussion}
This paper contributes to a series of works whose goal is to understand the limits
of image recovery using an invariant-based approach to solve multi-target
detection problems. 
After the challenges of particle picking were identified and a detection limit was proven,
a promising direction recalled autocorrelation analysis for image recovery that did not
rely on first identifying signal location within a measurement. This paper shows that the
estimation of the target image is possible theoretically in one-dimension and builds on this
intuition to empirically demonstrate that direct recovery is possible regardless of
noise level in two-dimensional settings with in-plane translations and rotations of the
target signal.

Future work includes extending Theorem~\ref{autos} to the two-dimensional
case, studying the high-dimensional regime where the size of the target image is
large (see for example~\cite{romanov2020multi,dou2022rates}), and exploring super-resolution
limits  in the MTD model~\cite{bendory2020super}.  Our ultimate goal is to
complete the program outlined in~\cite{bendory2018toward} and devise a
computational framework to recover a three-dimensional molecular structure
directly from micrographs.

\subsection*{Acknowledgments} 
NFM was supported in part by NSF DMS-1903015. TYL and AS 
were supported in part by Award Number FA9550-20-1-0266
from AFOSR, Simons Foundation Math+X Investigator Award, the Moore Foundation
Data-Driven Discovery Investigator Award, NSF BIGDATA Award IIS-1837992, NSF
DMS-2009753 and NIH/NIGMS 1R01GM136780-01 .  TB  was supported in part by
NSF-BSF grant no. 2019752,  ISF grant no. 1924/21,   BSF grant no. 2020159,
 and the Zimin Institute for Engineering Solutions
Advancing Better Lives. 

\bibliographystyle{plain}
\bibliography{ref}

\appendix
\section{Technical lemmas} 
\label{sec:technical_lemmas}
\subsection*{Proof of Lemma \ref{lem1} (Expectation)}
By the definition of $A_M$ and linearity of expectation, we have
$$
\mathbb{E} \left( A_M(y_1,y_2) \right) = \frac{1}{m} \sum_{x=1}^m \mathbb{E}
\left( M(x) M(x+y_1)
M(x+y_2) \right).
$$
Let $y_0 = 0$ for notational purposes.  By the definition of $M$, see
\eqref{M1d}, we have
$$
\mathbb{E} \left( A_M(y_1,y_2) \right) = \\ \frac{1}{m} \sum_{x=1}^m \mathbb{E}
\left( \prod_{k=0}^2 \left(\sum_{j=1}^p F_{\tau_j}(x - x_j +y_k ) +
\varepsilon(x + y_k) \right) \right).
$$
If the product in this expression is expanded, any terms with odd powers
of the noise term $\varepsilon(x)$ will have expectation zero. So, we only need to
consider terms where $\varepsilon$ does not appear, denoted $T_0(y_1,y_2)$, and
terms where $\varepsilon$ appears twice, denoted $T_2(y_1,y_2)$.  We have
\begin{multline*}
T_0(y_1,y_2) = \frac{1}{m} \sum_{x=1}^m  \sum_{j_1,j_2,j_3=1}^p \Big( \\
\mathbb{E} \left( F_{\tau_{j_1}}(x+x_{j_1}) F_{\tau_{j_2}}(x+y_1+x_{j_2})
F_{\tau_{j_3}}(x+y_2+x_{j_3}) \right) \Big).
\end{multline*}
By the separation condition, only terms where $j_1=j_2 = j_3$ are nonzero, so 
$$
T_0(y_1,y_2) = \frac{1}{m}  \sum_{j = 1}^p \Big( \\ \mathbb{E} \left(
\sum_{x=1}^m F_{\tau_{j}}(x+x_{j}) F_{\tau_{j}}(x+y_1+x_{j})
F_{\tau_{j}}(x+y_2+x_{j}) \right) \Big),
$$
with the sum over $x$ moved inside the expectation. Using the fact that
the circular shifts $\tau_j$ are uniformly random and that $F$ is supported on
$\{-n,\ldots,n-1\}$ gives
$$
T_0(y_1,y_2) = \frac{p}{m} \frac{1}{2n} \sum_{\tau=-n}^{n-1}  \sum_{x=-n}^{n-1}
\Big( \\
F_{\tau}(x) F_{\tau}(x+y_1) F_{\tau}(x+y_2) =  \frac{\gamma}{m} V_F(y_1,y_2)
\Big),
$$
where the final equality results from the definition of $V_F$ in \eqref{Teq} and
the definition of the density $\gamma = n p /m$.  It remains to consider the
terms $T_2(y_1,y_2)$ where two of the three noise terms $\varepsilon(x)$,
$\varepsilon(x+y_1)$, and $\varepsilon(x+y_2)$ appear.  The product of two of
these terms only has nonzero expectation if $y_1=0$,
$y_2=0$, or $y_1 =y_2$ so we have
\begin{multline*}
T_2(y_1,y_2) =  \frac{1}{m} \sum_{x=1}^m 
\sum_{j=1}^p
\mathbb{E}  \Big( \delta_0(y_1)
\varepsilon(x) \varepsilon(x+y_1)  F_{\tau_k}(x+x_{j}+y_2) + \\
 \delta_0(y_2)
\varepsilon(x) \varepsilon(x+y_2)  F_{\tau_k}(x+x_{j} + y_1) + \\
 \delta_{y_1}(y_2)
\varepsilon(x+y_1) \varepsilon(x+y_2)  F_{\tau_k}(x+x_{j}) \Big),
\end{multline*}
where $\delta_x(y) = 1$ when $x=y$ and $\delta_x(y) = 0$ otherwise.  By the
independence of the noise and the random cyclic shifts we have
$$
T_2(y_1,y_2) = \frac{p}{m} \sigma^2 2n T_F \left( \delta_0(y_1) + \delta_0(y_2)
+ \delta_{y_1}(y_2) \right),
$$
where $T_F$ denotes the mean of $F$, see \eqref{TF}. Adding $T_0(y_1,y_2)$ and
$T_2(y_1,y_2)$ gives the desired result:
$$
	\mathbb{E}( A_M(x_1,x_2)) = \frac{\gamma}{n} V_F(x_1, x_2) \\+ 2 \gamma T_F \sigma^2 (\delta_0 (x_1 -
	x_2) + \delta_0 (x_1) + \delta_0 (x_2)).
$$

\subsection*{Proof of Lemma \ref{lem1} (Variance)}
Let $\tilde{M}$ be an independent identically distributed copy of $M$. We can express the variance of $A_M(y_1, y_2)$ using $\tilde{M}$ as
$$
\Var \left( A_M(y_1,y_2) \right) = \\ \mathbb{E} \left( A_M(y_1,y_2) \left(
A_M(y_1,y_2) - A_{\tilde{M}}(y_1,y_2)  \right) \right).
$$
Expanding the right hand side gives
\begin{multline*}
\Var \left(A_M(y_1,y_2) \right) =  \frac{1}{m^2} \sum_{x,y=1}^m 
\mathbb{E} \left( M(x) M(x+y_1) M(x+y_2) \vphantom{\tilde{M}} \right.\\ 
 \left.\left(M(y) M(y+y_1) M(y+y_2) - \tilde{M}(y)
\tilde{M}(y+y_1) \tilde{M}(y+y_2) \right) \right).
\end{multline*}
By construction, the expectation of the terms in the sum is zero when $M(x)
M(x+y_1) M(x+y_2)$ and $M(y) M(y+y_1) M(y+y_2)$ are independent, which is the
case if
$$
x - y \not \in \{-n,\ldots,n-1\}.
$$
It follows that
\begin{multline*}
\Var \left( A_M(y_1,y_2) \right) =  
\frac{1}{m^2} \sum_{y=x-n + 1}^{x+n} \sum_{x=1}^m \mathbb{E} \left( M(x) M(x+y_1)
M(x+y_2) \vphantom{\tilde{M}} \right. \\ 
\left. \left(M(y) M(y+y_1) M(y+y_2) - \tilde{M}(y) \tilde{M}(y+y_1)
\tilde{M}(y+y_2) \right) \right).  
\end{multline*}
By Cauchy-Schwarz, it follows that
$$
\Var \left( A_M(y_1,y_2) \right) \le \\ \frac{1}{m^2} \sum_{y=x-n + 1}^{x+n}
\sum_{x=1}^m 2 \mathbb{E} \left(\left( M(x) M(x+y_1) M(x+y_2) \right)^2 \right).
$$
Given $|F| < F_{max}$ everywhere for some constant $F_{max}>0$, we can estimate
$$
\mathbb{E} \left( \left( M(x) M(x+y_1) M(x+y_2) \right)^2 \right) =
\mathcal{O}(\gamma F^6_{max} + \sigma^6)
$$
where $\sigma^2$ is the variance of the Gaussian noise. It follows that
$$
\Var \left( A_M(y_1,y_2) \right) \le \mathcal{O} \left( \frac{n}{m}(\gamma
F_{max}^6 + \sigma^6) \right),
$$
as was to be shown.

\end{document}